\documentclass[letterpaper, 10 pt, conference]{ieeeconf}  % Comment this line out if you need a4paper

\IEEEoverridecommandlockouts                              % This command is only needed if 
\usepackage{graphics} % for pdf, bitmapped graphics files
\usepackage{epsfig} % for postscript graphics files
\usepackage{mathptmx} % assumes new font selection scheme installed
\usepackage{times} % assumes new font selection scheme installed
\usepackage{amsmath} % assumes amsmath package installed
\usepackage{amssymb}  % assumes amsmath package installed
\usepackage{comment}
\usepackage{graphicx}
\usepackage{float}
\usepackage{subfig}
\usepackage{xcolor}

\usepackage{amsthm}
\usepackage{subcaption}

\newtheorem{remark}{Remark}
\newtheorem{theorem}{Theorem}

\newtheorem{corollary}{Corollary}
\usepackage{cite}
\usepackage{hyperref}

\title{\LARGE \bf
Mean-field control barrier functions for stochastic multi-agent systems%: applications in coverage and shepherding control with dangerous regions
}

\author{Cinzia Tomaselli, Gian Carlo Maffettone, Samy Wu Fung, Levon Nurbekyan and Mario di Bernardo%<-this % stops a space
\thanks{Cinzia Tomaselli, Gian Carlo Maffettone and Mario di Bernardo are with the Modeling and Engineering Risk and Complexity program, Scuola Superiore Meridionale, Naples, Italy
        {\tt\small \{c.tomaselli\}\{gc.maffettone\}@ssmeridionale.it}}%
\thanks{Samy Wu Fung is with the Department of Applied Mathematics and Statistics, Colorado School of Mines Golden, USA
        {\tt\small swufung@mines.edu}}%
\thanks{Levon Nurbekyan is with the Department of Mathematics, Emory University, Atlanta, USA
        {\tt\small lnurbek@emory.edu}}%
\thanks{Mario di Bernardo is also with the Department of Electrical Engineering and Information Technologies of the University of Naples Federico II, Naples, Italy
        {\tt\small mario.dibernardo@unina.it}}%
}

\begin{document}

\maketitle
\thispagestyle{empty}
\pagestyle{empty}

%%%%%%%%%%%%%%%%%%%%%%%%%%%%%%%%%%%%%%%%%%%%%%%%%%%%%%%%%%%%%%%%%%%%%%%%%%%%%%%%
\begin{abstract}
Many applications involving multi-agent systems require fulfilling safety constraints. Control barrier functions offer a systematic framework to enforce forward invariance of safety sets. Recent work extended this paradigm to mean-field scenarios, where the number of agents is large enough to make density-space descriptions a reasonable workaround for the curse of dimensionality. However, an open gap in the recent literature concerns the development of mean-field control barrier functions for Fokker-Planck (advection-diffusion) equations. In this work, we address this gap, enabling safe mean-field control of agents with stochastic microscopic dynamics. We provide bounded stability guarantees under safety corrections and corroborate our results through numerical simulations in two representative scenarios, coverage and shepherding control of multi-agent systems.
\end{abstract}

%%%%%%%%%%%%%%%%%%%%%%%%%%%%%%%%%%%%%%%%%%%%%%%%%%%%%%%%%%%%%%%%%%%%%%%%%%%%%%%%
\section{INTRODUCTION}
Continuum descriptions of the macroscopic behavior of multi-agent systems provide an effective modeling and control framework to circumvent the curse of dimensionality \cite{d2023controlling}. A representative example is swarm robotics, where modeling and controlling the density of the group, rather than tracking positions and velocities of all individual agents, offers a tractable approach \cite{sinigaglia2022density, elamvazhuthi2023density}.

In many applications, safety plays a crucial role, as dynamical systems may be required to avoid dangerous states while performing a control task. Control barrier functions (CBFs) are a well-established framework to enforce safety \cite{ames2019control}. The idea is to correct nominal control actions so that forward invariance of some safety set is ensured. These methods naturally extend to stochastic systems \cite{clark2021control} and multi-agent settings \cite{glotfelter2017nonsmooth}.

Recently, CBFs have been extended to continuum, mean-field settings. In particular,   mean-field CBFs (MF-CBFs) were introduced in \cite{fung2025mean}  in which safety-critical mean-field control problems are addressed for continuity/Liouville equations. A generalization of this technique to Banach spaces is presented in \cite{gao2026banach}.

An open gap in the recent works \cite{fung2025mean} and \cite{gao2026banach} is their focus on mean-field descriptions of deterministic systems. This results in macroscopic density dynamics that are not influenced by diffusion. However, many mean-field approaches start by considering stochastic microscopic dynamics, e.g. \cite{maffettone2025leader-follower, lama2024shepherding, albi2016invisible, ascione2023mean}.
Hence, in this work we extend CBFs for stochastic agents \cite{clark2021control} to mean-field settings, providing a generalization of the framework described in \cite{fung2025mean}. 

We validate our theoretical setup in two representative scenarios. First, we focus on a coverage problem \cite{cortes2004coverage} in the presence of dangerous zones. In this setting, a swarm of controllable agents is tasked with reaching a constant density profile in some domain of interest while avoiding dangerous zones. Then, we consider a shepherding control problem from robotics \cite{lama2024shepherding}, in which a set of leader agents (or herders) needs to collect and confine a set of follower agents (or targets) into a predefined goal region. 
These problems have applications ranging from swarm robotics and traffic \cite{dorigo2021swarm, siri2021freeway} to environmental monitoring \cite{kakalis2008robotic}, where direct control of all agents is infeasible.

To set the stage for our contribution, we briefly recall CBFs for stochastic systems \cite{clark2021control}. Consider the stochastic controlled dynamics
\begin{align}\label{eq:micro_dynamics}
    \mathrm{d}\mathbf{z}(t) = \mathbf{f}(\mathbf{z}(t), \mathbf{u}(t)) \mathrm{d}t + \sqrt{2D} \mathrm{d}\mathbf{W}(t),
\end{align}
where $\mathbf{z} \in\Omega\subset\mathbb{R}^d$ is the system's state, $\mathbf{f}$ models its deterministic dynamics, $\mathbf{u}$ is a control input, and $\mathbf{W}$ is a standard Wiener process with diffusion coefficient $D>0$. The control input $\mathbf{u}$ is typically chosen to perform some control task and then adjusted to guarantee safety. The system remains safe if, for all $t\geq 0$, $\mathbf{z}\in\mathcal{C}$ with probability 1, where $\mathcal{C}$ represents the safe operating region
\begin{align}
    \mathcal{C} = \{\mathbf{z} : h(\mathbf{z})\geq 0\}, \quad \partial\mathcal{C} = \{\mathbf{z} : h(\mathbf{z}) = 0\},
\end{align}
with $h$ being a locally Lipschitz scalar function. Provided that $\mathbf{z}(0) \in\mathcal{C}$, the control action $\mathbf{u}$ ensures safety if, $\forall t$,
\begin{align}\label{eq:micro-SCBF}
    \frac{\mathrm{d}}{\mathrm{d}t} h(\mathbf{z})= \nabla h(\mathbf{z}) \cdot \mathbf{f}(\mathbf{z}, \mathbf{u}) + D \Delta  h(\mathbf{z})\geq - \alpha(h(\mathbf{z})),
\end{align}
where $\alpha$ is a continuous strictly increasing function such that $\alpha(0) = 0$ (class-$\mathcal{K}$ function). One can then adjust $\mathbf{u}$ into its safety-corrected version $\mathbf{u}^*$ in real time by solving
\begin{equation}\label{eq:micro-QP}
\begin{aligned}
&\mathbf{u}^*
= \arg\min_{\mathbf{u}^*}  \Vert\mathbf{u}^*-\mathbf{u}\Vert^2\\
& \text{s.t.}\quad
\nabla h(\mathbf{z}) \cdot \mathbf{f}(\mathbf{z}, \mathbf{u}) + D \, \Delta h(\mathbf{z})\geq - \alpha(h(\mathbf{z})).
\end{aligned}
\end{equation}
As noted in \cite{fung2025mean} for the deterministic case, CBFs are prone to the curse of dimensionality in multi-agent settings. This motivates more scalable formulations based on mean-field approximations.

The rest of the paper is organized as follows. %In Sec. \ref{sec:backgorunf} we go through CBFs for stochastic systems, the starting point of our setting. 
In Sec. \ref{sec: mf_cbf}, we introduce mean-field control barrier functions for advection-diffusion equations, representing the mean-field approximation of stochastic agents. In Sec. \ref{sec:kwo}, we introduce the metric we use to describe our safety sets. Finally, in Secs. \ref{sec:coverage} and \ref{sec:shepherding}, we apply our framework to two characteristic settings, coverage and shepherding. Conclusions in Sec. \ref{sec:conclusions} close our work.

\section{Mean Field Control Barrier Function for advection-diffusion Systems}
\label{sec: mf_cbf}
For collections of dynamical systems of the form \eqref{eq:micro_dynamics}, mean-field models offer more tractable and compact formulations. In this macroscopic setting, the focus shifts from tracking the states of a large number of agents, towards the dynamics of their probability density distribution.

Consider the evolution of a population density $\rho(t,\mathbf{x})$ in a domain $\Omega \subset \mathbb{R}^d$, equipped with boundary conditions that ensure mass conservation, such as periodic or reflective boundaries. In the mean-field limit, the empirical measure of agents governed by \eqref{eq:micro_dynamics} converges to the density, $\rho$, which evolves according to the advection--diffusion equation
\begin{equation}
\partial_t \rho(t,\mathbf{x})
+ \nabla \cdot \big(\rho(t,\mathbf{x})\, \mathbf{w}(t,\mathbf{x},\mathbf{u}(t,\mathbf{x}))\big)
= D\, \Delta \rho(t,\mathbf{x}),
\label{eq:PDE}
\end{equation}
where $\mathbf{w}$ denotes a velocity field parameterized by a Lipschitz continuous control input $\mathbf{u}$, and $D>0$ is the diffusion coefficient. Here, we operate within mean-field control/mean-field game frameworks, where the collective behavior of large populations of agents is modeled and controlled at the level of density distributions
\cite{lasry2007mean, agrawal2022random, lin2021alternating}.

Let $\mathcal{H}$ be a continuously differentiable functional on the space of densities, and define the associated safe set $\mathcal{C}$ as
\begin{equation}
\mathcal{C}
=
\left\{ \rho : \mathcal{H}(t,\rho) > 0 \right\},
\qquad
\partial\mathcal{C}
=
\left\{ \rho : \mathcal{H}(t,\rho)=0 \right\}.
\label{eq: safety_set_C}
\end{equation}
Following the MF-CBF paradigm, assuming $\mathcal{H}(t_0,\rho)\geq 0$ at initial time $t_0$, forward invariance of $\mathcal{C}$ is ensured by enforcing the mean-field control barrier condition
\begin{equation}
\frac{\mathrm{d}}{\mathrm{d}t} \mathcal{H}(t,\rho(t, \mathbf{x}))
\;\ge\;
-\alpha\!\left(\mathcal{H}(t,\rho(t, \mathbf{x}))\right),
\label{eq:CBF_condition}
\end{equation}
where $\alpha(\cdot)$ is a class-$\mathcal{K}$ function. The explicit time dependence of $\mathcal{H}$ allows this formulation to accommodate time-varying (e.g., moving) dangerous zones.

\begin{theorem}
Consider a density distribution $\rho(t,\mathbf{x})$ evolving according to the advection--diffusion equation~\eqref{eq:PDE} over a domain $\Omega$ subject to boundary conditions ensuring mass conservation, that is $\partial_t\left(\int_\Omega \rho\,\mathrm{d}\mathbf{x}\right)=0$. 
Let $\mathcal{H}(t, \rho)$ be a continuously differentiable functional defining the safe set $\mathcal{C}$ as in~\eqref{eq: safety_set_C}, and denote by $\delta_\rho \mathcal{H}$ its Fr\'{e}chet derivative with respect to $\rho$. Let $\alpha$ be a class-$\mathcal{K}$ function.
Then, if $\rho(t_0,\mathbf{x})\in \mathcal{C}$ at the initial time $t_0$, any control input $\mathbf{u}$ selected from the set
\begin{equation}
\begin{aligned}
\mathcal{K}_{CBF}(t,\rho)
=
\Big\{ \mathbf{u}:& \;
\int_{\Omega}
\nabla \delta_\rho \mathcal{H}(t,\rho)\cdot \rho\,\mathbf{w}(t,\mathbf{x},\mathbf{u})\, \mathrm{d}\mathbf{x}
\ge \\ &
- D \!\int_{\Omega}
\Delta\!\left(\delta_\rho \mathcal{H}(t,\rho)\right)\rho\, \mathrm{d}\mathbf{x}
\\&- \partial_t \mathcal{H}(t,\rho) 
- \alpha(\mathcal{H}(t,\rho))
\Big\}
\label{eq:safe_set_general}
\end{aligned}
\end{equation}
guarantees forward invariance of the safe set $\mathcal{C}$.
\end{theorem}
\begin{proof}
We compute the total time derivative of the functional $\mathcal{H}(t,\rho)$ appearing in \eqref{eq:CBF_condition}. 
\begin{comment}
We now derive an explicit expression for the total time derivative appearing in \eqref{eq:CBF_condition}.
\begin{theorem}
Let $\mathcal{H}$ be a continuously differentiable functional and let $\rho(t,\mathbf{x})$ satisfy the advection--diffusion equation~\eqref{eq:PDE} over a domain $\Omega$ ensuring mass conservation. Then, the total time derivative of $\mathcal{H}$ along the density evolution is given by
\begin{equation}
\begin{split}
\frac{d}{dt} \mathcal{H}(t,\rho(t))
= &
\int_{\Omega}
\nabla \delta_\rho \mathcal{H}(t,\rho)\cdot \rho(t,\mathbf{x})\,
\mathbf{w}(t,\mathbf{x},\mathbf{u}) \, \mathrm{d}\mathbf{x}  \\ & +
D\int_{\Omega}
\Delta\!\left(\delta_\rho \mathcal{H}(t,\rho)\right)\, \rho(t,\mathbf{x}) \, \mathrm{d}\mathbf{x}
+ \partial_t \mathcal{H}(t,\rho).
\end{split}
\label{eq:Hdot_final}
\end{equation}
\end{theorem}
\end{comment}

Using the chain rule for time-dependent functionals,
\begin{equation}
\frac{\mathrm{d}}{\mathrm{d}t}\mathcal{H}(t,\rho(t, \mathbf{x}))
=
\int_{\Omega}
\delta_\rho \mathcal{H}(t,\rho)\,
\partial_t \rho(t,\mathbf{x})\, \mathrm{d}\mathbf{x}
+
\partial_t \mathcal{H}(t,\rho).
\label{eq:chainrule}
\end{equation}

Substituting the dynamics \eqref{eq:PDE} into
\eqref{eq:chainrule} yields
\begin{equation}
\begin{aligned}
\frac{d}{dt}\mathcal{H}(t,\rho(t, \mathbf{x}))
=&
- \int_{\Omega}
\delta_\rho \mathcal{H}(t,\rho)\,
\nabla \cdot (\rho \mathbf{w})\, \mathrm{d}\mathbf{x}
\\ & \quad
+ D \int_{\Omega}
\delta_\rho \mathcal{H}(t,\rho)\,
\Delta \rho \, \mathrm{d}\mathbf{x}
+ \partial_t \mathcal{H}(t,\rho).
\end{aligned}
\label{eq:substituted}
\end{equation}
 
For the advection term in \eqref{eq:substituted}, integration by parts together with boundary conditions gives
\begin{equation}
- \int_{\Omega}
\delta_\rho \mathcal{H}\,
\nabla \cdot (\rho \mathbf{w})\, \mathrm{d}\mathbf{x}
=
\int_{\Omega}
\nabla \delta_\rho \mathcal{H}\cdot \rho \mathbf{w} \, \mathrm{d}\mathbf{x}.
\label{eq:int_by_parts1}
\end{equation}
For the diffusion term in \eqref{eq:substituted}, applying integration by parts twice (or equivalently,
the first Green identity) yields
\begin{equation}
D \int_{\Omega}
\delta_\rho \mathcal{H}\, \Delta \rho \, \mathrm{d}\mathbf{x}
=
D \int_{\Omega}
\Delta\!\left(\delta_\rho \mathcal{H}\right)\, \rho \, \mathrm{d}\mathbf{x}.
\label{eq:int_by_parts2}
\end{equation}
Combining \eqref{eq:int_by_parts1} and \eqref{eq:int_by_parts2} with
\eqref{eq:substituted}, we obtain
\begin{equation}
\begin{aligned}
\frac{\mathrm{d}}{\mathrm{d}t} \mathcal{H}(t,\rho(t, &\mathbf{x}))
= 
\int_{\Omega}
\nabla \delta_\rho \mathcal{H}(t,\rho)\cdot \rho(t,\mathbf{x})\,
\mathbf{w}(t,\mathbf{x},\mathbf{u}) \, \mathrm{d}\mathbf{x} \\ & 
+ D\int_{\Omega}
\Delta\!\left(\delta_\rho \mathcal{H}(t,\rho)\right)\, \rho(t,\mathbf{x}) \, \mathrm{d}\mathbf{x}
+ \partial_t \mathcal{H}(t,\rho)
\end{aligned}
\end{equation}

Enforcing \eqref{eq:CBF_condition} then yields the MF-CBF constraint
\begin{equation}
\begin{aligned}
& \int_{\Omega}
\nabla \delta_\rho \mathcal{H}(t,\rho)\cdot
\rho(t,\mathbf{x})\,\mathbf{w}(t,\mathbf{x},\mathbf{u}) \, \mathrm{d}\mathbf{x} 
\;\ge\; \\
& - D \int_{\Omega}
\Delta\!\left(\delta_\rho \mathcal{H}(t,\rho)\right)\,
\rho(t,\mathbf{x}) \, \mathrm{d}\mathbf{x}
-\partial_t \mathcal{H}(t,\rho)
- \alpha\!\left(\mathcal{H}(t,\rho)\right),
\end{aligned}
\label{eq:MF_CBF_final}
\end{equation}
which defines the set of safe control inputs in \eqref{eq:safe_set_general}, concluding the proof.
\end{proof}

Finally, given a nominal control input $\mathbf{u}(t,\mathbf{x})$, its safety-adjusted counterpart $\mathbf{u}^*(t,\mathbf{x})$ can be computed by solving:
\begin{equation}
\begin{aligned}
\mathbf{u}^*(t,\mathbf{x})
=& \arg\min_{\mathbf{u}'}  \mathcal{L}[\mathbf{u}'(t, \mathbf{x} )-\mathbf{u}(t, \mathbf{x} )]\\
& \text{s.t.}\quad
\mathbf{u}' \in \mathcal{K}_{\mathrm{CBF}}(t, \rho),
\end{aligned}
\label{eq:qp_general}
\end{equation}
where $\mathcal{L}$ is a quadratic functional cost such as the squared $\mathcal{L}^2(\Omega)$-norm. 

We note that \eqref{eq:safe_set_general} closely resembles its microscopic counterpart \eqref{eq:micro-SCBF}. Specifically, both the constraints involve a term related to the deterministic drift and a term related to diffusion arising from stochasticity.

\subsection{Bounded convergence with MF-CBFs}
We now analyze how MF-CBFs affect convergence properties in the representative scenario of a density control problem.
Consider \eqref{eq:PDE} and assume that, in the absence of dangerous regions, $\mathbf{u}$ is chosen such that $\lim\limits_{t\to\infty}\rho(t, \mathbf{x}) = \bar{\rho}(\mathbf{x})$. We denote by $\mathbf{w}(\mathbf{u}^*)$ the safety-corrected drift of \eqref{eq:PDE}.

\begin{theorem}
    If $\nabla\cdot\mathbf{w}(\mathbf{u}^*) \in\mathcal{L}^\infty(\Omega)$, $\nabla\cdot(\bar{\rho}\mathbf{w}(\mathbf{u}^*))\in\mathcal{L}^2(\Omega)$, $\Delta\bar{\rho}\in\mathcal{L}^2(\Omega)$ and $D > \frac{\Vert \nabla\cdot\mathbf{w}(\mathbf{u}^*)\Vert_\infty}{2}$,
    system \eqref{eq:PDE} subject to MF-CBFs correction is globally bounded stable.
\end{theorem}
\begin{proof}
    We consider the Lyapunov functional $V = \frac{1}{2}\Vert e\Vert_2^2$, where $e = \bar{\rho} - \rho$. Its time differentiation yields
    \begin{align}\label{eq:Vt}
        \frac{\mathrm{d}}{\mathrm{d}t}V(t) = \int_\Omega e(t, \mathbf{x})\partial_t e(t, \mathbf{x})\,\mathrm{d}\mathbf{x}.
    \end{align}
    The error dynamics obeys
    \begin{multline}\label{eq:err_dyn}
        \partial_te(t, \mathbf{x}) = \nabla\cdot (\bar{\rho}(\mathbf{x}) \mathbf{w}(t, \mathbf{x}, \mathbf{u}^*(t, \mathbf{x}))) \\- \nabla\cdot (e(t, \mathbf{x}) \mathbf{w}(t, \mathbf{x}, \mathbf{u}^*(t, \mathbf{x}))) -D\Delta\bar{\rho}(\mathbf{x}) + D\Delta e(t, \mathbf{x}).
    \end{multline}
    Substituting \eqref{eq:err_dyn} into \eqref{eq:Vt} yields (dropping explicit dependence from time and space)
    \begin{equation}\label{eq:Vt2}
    \begin{aligned}
        \frac{\mathrm{d}}{\mathrm{d}t}V ={}& \int_\Omega e\nabla\cdot(\bar{\rho}\mathbf{w}(\mathbf{u}^*))\,\mathrm{d}\mathbf{x}
        - \int_\Omega e\nabla\cdot(e\mathbf{w}(\mathbf{u}^*))\,\mathrm{d}\mathbf{x} \\
        &\; - D\int_\Omega e\Delta \bar{\rho}\,\mathrm{d}\mathbf{x}
        + D\int_\Omega e\Delta e\,\mathrm{d}\mathbf{x}.
    \end{aligned}
    \end{equation}
    We can establish the following bounds:   
    \begin{subequations}\label{eq:bounds}
    \begin{equation}\label{eq:bounds_a}
    \begin{aligned}
    \left\lvert \int_\Omega e \nabla\cdot(\bar{\rho}\mathbf{w}(\mathbf{u}^*)) \,\mathrm{d}\mathbf{x} \right\rvert
    &\leq \left\| e \nabla\cdot(\bar{\rho}\mathbf{w}(\mathbf{u}^*)) \right\|_1 \\
    &\leq \|e\|_2 \left\| \nabla\cdot(\bar{\rho}\mathbf{w}(\mathbf{u}^*)) \right\|_2 \\
    &= \sqrt{2}\left\| \nabla\cdot(\bar{\rho}\mathbf{w}(\mathbf{u}^*)) \right\|_2 \sqrt{V},
    \end{aligned}
    \end{equation}    
    \begin{equation}\label{eq:bounds_b}
    \begin{aligned}
    \int_\Omega e \nabla\cdot(e\mathbf{w}(\mathbf{u}^*)) \,\mathrm{d}\mathbf{x}
    &= -\int_\Omega \nabla e \cdot e\mathbf{w}(\mathbf{u}^*) \,\mathrm{d}\mathbf{x} \\
    &= -\frac{1}{2}\int_\Omega \nabla(e^2)\cdot\mathbf{w}(\mathbf{u}^*) \,\mathrm{d}\mathbf{x} \\
    &= \frac{1}{2}\int_\Omega e^2 \nabla\cdot\mathbf{w}(\mathbf{u}^*) \,\mathrm{d}\mathbf{x} \\
    &\leq \frac{1}{2}\left\| e^2 \nabla\cdot\mathbf{w}(\mathbf{u}^*) \right\|_1 \\
    &\leq \left\| \nabla\cdot\mathbf{w}(\mathbf{u}^*) \right\|_\infty V,
    \end{aligned}
    \end{equation}
    \begin{equation}\label{eq:bounds_c}
    \begin{aligned}
    D\left\lvert \int_\Omega e \Delta \bar{\rho} \,\mathrm{d}\mathbf{x} \right\rvert
    &\leq D\left\| e \Delta \bar{\rho} \right\|_1 \\
    &\leq D\|e\|_2 \left\| \Delta \bar{\rho} \right\|_2 = \sqrt{2}\,D \left\| \Delta \bar{\rho} \right\|_2 \sqrt{V},
    \end{aligned}
    \end{equation}
    \begin{equation}\label{eq:bounds_d}
    \begin{aligned}
    D\int_\Omega e \Delta e \,\mathrm{d}\mathbf{x}
    &= D\int_\Omega e \nabla\cdot(\nabla e) \,\mathrm{d}\mathbf{x} \\
    &= -D\int_\Omega \|\nabla e\|^2 \,\mathrm{d}\mathbf{x} \leq -2DV,
    \end{aligned}
    \end{equation}
    \end{subequations}
    where we use H$\mathrm{\Ddot{o}}$lder's inequality \cite{axler2020measure} and integration by parts in higher dimensions (boundary terms vanish due to boundary conditions), and the Poincarè-Wirtinger inequality~\cite{heinonen2001lectures}.

   Substituting \eqref{eq:bounds} into \eqref{eq:Vt2} yields
    \begin{align}
        \frac{\mathrm{d}}{\mathrm{d}t}V \leq -aV+b\sqrt{V},
    \end{align}
    where $a = 2D - \Vert \nabla\cdot\mathbf{w}(\mathbf{u}^*)\Vert_\infty$ (which is positive under the theorem hypothesis) and $b = \sqrt{2}(\Vert \nabla\cdot(\bar{\rho}\mathbf{w}(\mathbf{u}^*))\Vert_2+\Vert \Delta \bar{\rho}\Vert_2)$. The scalar dynamical system bounding $\frac{\mathrm{d}}{\mathrm{d}t}V$ has a globally stable equilibrium point at $(b/a)^{2}$, proving the claim by the comparison principle. Note that in the absence of dangerous zones, this equilibrium converges to 0.
\end{proof}

\begin{remark}
The condition $D > \frac{\|\nabla \cdot \mathbf{w}(\mathbf{u}^*)\|_\infty}{2}$ involves the safety-corrected drift, which itself depends on $D$ through~\eqref{eq:qp_general}. However, since the safety correction is computed first for a given $D$, the condition can be verified aposteriori.
\end{remark}

\section{Kernel-weighted Overlap as Mean-Field Control Barrier Function}\label{sec:kwo}

The MF-CBF framework developed in Section II applies to any continuously differentiable functional $\mathcal{H}$; in the remainder of this paper, we adopt a specific construction based on kernel-weighted overlaps, which provides a natural and computationally tractable way to quantify the interaction between the population density and dangerous regions. Specifically, we define a kernel-weighted overlap between probability measures and show that bounding this quantity suffices to limit the mass of the population density inside a dangerous region. To this end, assume that $O\subset \Omega$ is a dangerous region and $\rho_O$ is a probability measure on $O$. We call $O^c$ the complement of $O$ w.r.t. $\Omega$. The \textit{kernel-weighted overlap} between probability measures $\rho$ and $\rho_O$ is defined as
 \begin{equation}
R_k(\rho, \rho_O)=\iint_{\Omega\times\Omega}k(\mathbf{x},\mathbf{y})\rho(\mathbf{x})\rho_O(\mathbf{y})\,\mathrm{d}\mathbf{x}\mathrm{d}\mathbf{y},
\end{equation}
where $k(\mathbf{x},\mathbf{y})$ is a positive-definite kernel~\cite{gretton2012kernel} measuring some type of similarity between states $\mathbf{x},\mathbf{y}$. A typical choice is to use the Gaussian kernel  $k(\mathbf{x}, \mathbf{y}) = \exp\!\left(-\frac{\|\mathbf{x} - \mathbf{y}\|^2}{2\sigma^2}\right)$.

Furthermore, consider the quantity
\begin{equation}\label{eq:B(alpha)}
    \begin{split}
        B(\mu)=& \inf_{\int_O \rho(\mathbf{x})\mathrm{d}\mathbf{x} = \mu } R_k(\rho,\rho_O),
    \end{split}
\end{equation}
where $\mu \in [0,1]$. $B(\mu)$ represents the minimum kernel-weighted overlap achievable by any density that places exactly mass $\mu$ inside $O$. In what follows, we prove that $B$ is a strictly increasing function, so that the inequality $B(\mu)<\epsilon$ guarantees that $\int_O \rho(\mathbf{x})dx<  \mu$ for a suitably chosen $\epsilon$. Before we proceed, we introduce the notation
\begin{equation}\label{eq:C({x})}
    C(\mathbf{x})=\int_\Omega k(\mathbf{x},\mathbf{y})\rho_O(\mathbf{y})\mathrm{d}\mathbf{y},\quad \mathbf{x}\in \Omega.
\end{equation}

\begin{theorem}\label{thm:B(alpha)}
    The function $B$ admits the explicit representation
    \begin{equation}\label{eq:B(alpha)_formula}
        B(\mu)=\mu \cdot \inf_{\mathbf{x}\in O} C(\mathbf{x})+(1-\mu) \inf_{\mathbf{y}\in O^c} C(\mathbf{y}),\quad 0\leq \mu \leq 1.
    \end{equation}
    Hence, if $\inf_{\mathbf{x}\in O} C(\mathbf{x})> \inf_{\mathbf{y}\in O^c} C(\mathbf{y})$, then $B(\mu)$ is strictly increasing and
    \[
    R_k(\rho,\rho_O)< B(\mu) \Longrightarrow \int_O \rho(\mathbf{x})\mathrm{d}\mathbf{x} < \mu.
    \]
\end{theorem}
\begin{proof}
Let $\mu \in [0,1]$ and $\rho$ be such that $\int_O\rho(\mathbf{x})\mathrm{d}\mathbf{x}=\mu$, and $\rho_{in}(\mathbf{x})=\rho(\mathbf{x} \mid O)$, and $\rho_{out}(\mathbf{x})=\rho(\mathbf{x} \mid O^c)$. Then we have that
\[
\rho=\mu \rho_{in}+(1-\mu)\rho_{out},
\]
and
\[
\begin{aligned}
    R_k(\rho,\rho_O)&=R_k(\mu \rho_{in}+(1-\mu)\rho_{out},\rho_O)\\
    &=\mu R_k(\rho_{in},\rho_O)+(1-\mu)R_k(\rho_{out},\rho_O)\\
    &= \mu \int_O C(\mathbf{x}) \rho_{in}(\mathbf{x})\mathrm{d}\mathbf{x}+(1-\mu)\int_{O^c} C(\mathbf{y}) \rho_{out}(\mathbf{y})\mathrm{d}\mathbf{y}\\
    &\geq \mu \cdot \inf_{\mathbf{x}\in O} C(\mathbf{x})+(1-\mu)\cdot \inf_{\mathbf{y}\in O^c} C(\mathbf{y}).
\end{aligned}
\]
Hence, we have that 
\[
B(\mu) \geq \mu \cdot \inf_{\mathbf{x}\in O} C(\mathbf{x})+(1-\mu)\cdot \inf_{\mathbf{y}\in O^c} C(\mathbf{y}).
\]
Furthermore, let $(z_n)\subset O $ and $(w_n) \subset O^c$ be such that
    \[
    \lim_{n \to \infty}C(z_n) = \inf_{\mathbf{x}\in O} C(\mathbf{x}),\qquad \lim_{n \to \infty}C(w_n) = \inf_{\mathbf{y}\in O^c} C(\mathbf{y}).
    \]
    Then for $\rho_n=\mu \delta_{z_n}+(1-\mu)\delta_{w_n}$ (or suitable smooth approximation of these) we have that
    \begin{equation*}
        \begin{split}
            R_k(\rho_n,\rho_O)=&~\mu \cdot C(z_n)+(1-\mu)\cdot C(w_n)\\
            \longrightarrow & ~\mu \inf_{\mathbf{x}\in O} C(\mathbf{x})+ (1-\mu) \inf_{\mathbf{y}\in O^c}C(\mathbf{y}),
        \end{split}
    \end{equation*}
    and so
    \[
    B(\mu) \leq  \mu \cdot \inf_{\mathbf{x}\in O} C(\mathbf{x})+(1-\mu)\cdot \inf_{\mathbf{y}\in O^c} C(\mathbf{y}),
    \]
    which proves~\eqref{eq:B(alpha)_formula}. The remaining statements follow readily from~\eqref{eq:B(alpha)_formula}.
\end{proof}

\begin{corollary}\label{crl:A(alpha)andB(alpha)}
Under the hypotheses of Theorem~\ref{thm:B(alpha)}, for every $\mu \in (0,1)$ there exists $\epsilon>0$ such that $R_k(\rho,\rho_O)<\epsilon$ implies $\int_O \rho(\mathbf{x})\mathrm{d}\mathbf{x} <\mu$.
\end{corollary}
\begin{proof}
    Taking $\epsilon=B(\mu)$, the result follows from Theorem~\ref{thm:B(alpha)}.
\end{proof}

\begin{remark}
    Theorem~\ref{thm:B(alpha)} and Corollary~\ref{crl:A(alpha)andB(alpha)} establish that the kernel-weighted overlap is a valid quantity for limiting the mass of the population inside dangerous regions. In practice, as shown in Sections~\ref{sec:coverage} and~\ref{sec:shepherding}, the bound $\epsilon$ from Corollary~\ref{crl:A(alpha)andB(alpha)} is conservative, and looser bounds suffice.
\end{remark}

In practice, the threshold $\varepsilon$ need not satisfy the conservative bound $\varepsilon = B(\mu)$ from Corollary~\ref{crl:A(alpha)andB(alpha)}. Instead, it can be tuned as a design parameter: smaller values of $\varepsilon$ enforce a tighter safety margin between the population and the dangerous regions, at the cost of a more restrictive feasible set $\mathcal{K}_{CBF}(t,\rho)$, which may lead to larger deviations from the nominal control. The choice of $\varepsilon$ thus reflects a trade-off between safety stringency and control performance.

\begin{figure}[t]
    \centering
    %--- SUBFIGURES ---
    \subfloat[]
{
    \includegraphics{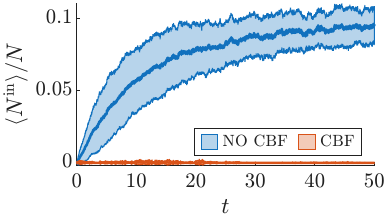}
    \label{fig:temporal_Nin}
}

\subfloat[]
{
    \includegraphics{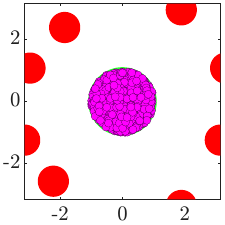}
    \label{fig: start_heat}
}
\subfloat[]
{
    \includegraphics{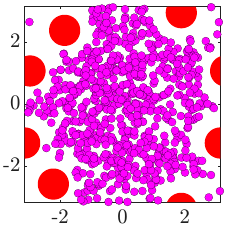}
    \label{fig: end_heat}
}
   \caption{{\em Safe coverage control.} Effect of MF-CBF enforcement on agents entering dangerous regions.  
    (a) Temporal evolution of the fraction of agents entering the dangerous regions averaged over 50 simulations. 
    (b) Initial and (c) final spatial configurations of the population (magenta disks) when MF-CBF is enforced.}
    \label{fig:N_int}
\end{figure}

\begin{figure}[t]
    \centering
    %--- SUBFIGURES ---
    \subfloat[]
{
    \includegraphics{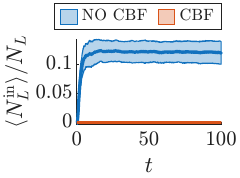}
    \label{fig:temporal_Nl}
}
\subfloat[]
{
    \includegraphics{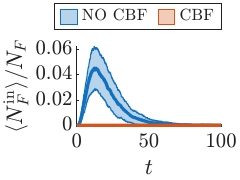}
    \label{fig:temporal_Nf}
}

\subfloat[]{
    \includegraphics{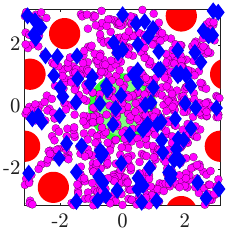}
    \label{fig:initial}
}
\subfloat[]{
    \includegraphics{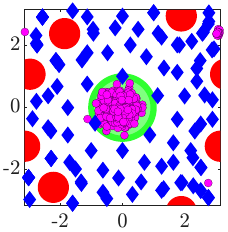}
    \label{fig:final}
}
   \caption{{\em Safe shepherding control.} Temporal evolution of the average fraction of leaders (a) and followers (b) entering the dangerous regions.
   (c) Initial and (d) final spatial configurations of the leaders (blue diamonds) and followers (magenta disks) populations when MF-CBF is enforced. At final time, a fraction $0.967$ of followers lies within the goal region}
    \label{fig: N_int_lf}
\end{figure}

\section{Applications and Numerical Validation}\label{sec:validation}

We validate the proposed MF-CBF framework in two representative application scenarios: coverage control and shepherding. Both are set in the periodic domain $\Omega=[-\pi,\pi]^2$ and involve static dangerous regions modeled as the union of disks with uniform density $\rho_O:\Omega\to\mathbb{R}_{\geq 0}$, vanishing outside their support. In both cases, the barrier function is defined as $\mathcal{H}(\rho) = \varepsilon - R_k(\rho,\rho_O)$, using a periodic Gaussian kernel with bandwidth $\sigma$, and safety is enforced by solving~\eqref{eq:qp_general}.
In all simulations, we choose $\alpha(s) = \gamma s$ with $\gamma = 0.1$, corresponding to a linear class-$\mathcal{K}$ function.

To approximate the density dynamics, we employ a particle-based representation~\cite{del2013mean}, in which the agents' positions $\mathbf{x}_k(t)$ ($k=1,\ldots,N$) are updated through the Euler-Maruyama scheme. The density $\rho(t,\mathbf{x})$ is reconstructed from agent positions via the empirical measure
\begin{equation}\label{eq:empirical_measure}
    \hat{\rho}(t,\mathbf{x}) = \frac{1}{N}\sum_{k=1}^{N}\delta\!\left(\{\mathbf{x},\mathbf{x}_k(t)\}_\Omega\right),
\end{equation}
with $\{\mathbf{x},\mathbf{y}\}_\Omega$ denoting the relative position wrapped on the periodic domain, computed element-wise. Specifically, for each component $i = 1,2$, due to the $2\pi$-periodicity of the domain,
\begin{equation} 
\begin{aligned}
    \{x_i, y_i\}_\Omega
= & \min\big\{
\,|x_i - y_i|,\; 2\pi- |x_i - y_i|
\big\}\,\operatorname{sign}(x_i - y_i) \\ 
& \cdot \operatorname{sign}(\pi-|x_i - y_i|)
\end{aligned}
\end{equation}
The corresponding wrapped distance is therefore given by 
$\|\{\mathbf{x}, \mathbf{y}\}_\Omega \|
    =
    \sqrt{
      \sum_{i=1}^2 
      \bigl( \{x_i, y_i\}_\Omega \bigr)^2 }$.
      
 Density-dependent quantities, including $\mathcal{H}(\rho)$ and $\delta_\rho\mathcal{H}$, are also evaluated using the empirical approximation \eqref{eq:empirical_measure}. The decision variable in~\eqref{eq:qp_general} becomes the stacked vector $[\mathbf{u}_1'(t)^\top,\ldots,\mathbf{u}_N'(t)^\top]^\top\in\mathbb{R}^{2N}$, and choosing $\mathcal{L}$ as the squared $\mathcal{L}^2(\Omega)$-norm yields the finite-dimensional quadratic program
\begin{equation}\label{eq:empirical_cost}
    \mathcal{L}^2[\mathbf{u}'-\mathbf{u}] \approx \frac{1}{N}\sum_{k=1}^{N}\left\|\mathbf{u}_k'(t)-\mathbf{u}_k(t)\right\|^2.
\end{equation}
For each scenario, we compare the cases with and without MF-CBF over 50 independent simulations. Our analysis focuses on the fraction of agents entering the dangerous regions, denoted by $N^{\mathrm{in}}/N$.

\subsection{Safe coverage control}\label{sec:coverage}

We consider a coverage problem~\cite{cortes2004coverage} in which $N=720$ agents must achieve a constant density profile while avoiding dangerous regions. The density evolves according to
\begin{equation}\label{eq:coverage_PDE}
    \partial_t\rho(t,\mathbf{x}) + \nabla\cdot\left[\rho(t,\mathbf{x})\,\mathbf{u}(t,\mathbf{x})\right] = D\,\Delta\rho(t,\mathbf{x}),
\end{equation}
where $\mathbf{u}$ is the control velocity field and $D>0$ is the diffusion coefficient. We exploit the stochastic behavior of the agents: when $\mathbf{u}=0$, \eqref{eq:coverage_PDE} reduces to the heat equation, which converges globally and asymptotically to the constant density profile under periodic boundary conditions. To avoid dangerous zones, $\mathbf{u}$ is adjusted via the MF-CBF constraint~\eqref{eq:qp_general}.

Simulations use $D=0.05$, final time $T=50$, time step $\delta t = 0.01$, kernel bandwidth $\sigma=0.2$, and barrier threshold $\varepsilon=0.01$. 

The population is initially distributed uniformly in the disk 
$\{\mathbf{x} \in \mathbb{R}^2 : \|\mathbf{x}\| \leq 1\}$ centered at the origin. 
Dangerous regions consist of five disks of radius $r_O = 0.5$. 
The center of each disk is at least $\lambda = 0.75$ away from the boundary of the initial population support, 
and the minimum distance between the boundaries of any two disks is $0.25$. 
Hence, the barrier constraints are satisfied at the initial time.

Fig.~\ref{fig:temporal_Nin} shows the average fraction of agents entering the dangerous regions. Without MF-CBF, this percentage grows rapidly and remains large. With MF-CBF, it stays close to zero throughout. The initial and final spatial configuration of the population is shown in Figs.~\ref{fig: start_heat} and \ref{fig: end_heat}. The latter confirms that at $t=50$, no agent lies in the dangerous regions.

\subsection{Safe shepherding control}\label{sec:shepherding}

We consider a shepherding problem~\cite{lama2024shepherding} in which leaders must confine followers into a goal region while both populations avoid dangerous areas. Most existing approaches to shepherding in the presence of dangerous regions rely on the assumption of cohesion among followers, which simplifies the confinement task by preserving group compactness. To the best of our knowledge, only few works address this problem without cohesion assumptions, including a reinforcement learning approach~\cite{punzo2025decentralized} and a heuristic strategy~\cite{tomaselli_icra2026}, both operating at the agent level. Here, we tackle this problem in a large-scale, density-based setting following~\cite{di2025continuification}, where analytical and computational tractability are facilitated by the continuum formulation.

Following \cite{maffettone2025leader-follower}, leaders and followers have densities $\rho_i:\mathbb{R}_{\geq 0}\times\Omega\to\mathbb{R}_{\geq 0}$, $i\in\{L,F\}$, governed by the coupled PDEs
\begin{align}
    \partial_t\rho_L(t,\mathbf{x}) + \nabla\cdot\left[\rho_L(t,\mathbf{x})\,\mathbf{u}(t,\mathbf{x})\right] &= 0, \label{eq:leaders_PDE}\\
    \partial_t\rho_F(t,\mathbf{x}) + \nabla\cdot\left[\rho_F(t,\mathbf{x})\,\mathbf{v}_F(t,\mathbf{x})\right] &= D\,\Delta\rho_F(t,\mathbf{x}), \label{eq:followers_PDE}
\end{align}
where $\mathbf{u}$ is the leader control field, $D>0$ is the follower diffusion coefficient, and $\mathbf{v}_F = (\mathbf{g}*\rho_L)(t,\mathbf{x})$ models repulsive leader--follower interaction through the kernel
\begin{equation}\label{eq:repulsive_kernel}
    \mathbf{g}(\{\mathbf{x},\mathbf{y}\}_\Omega) = \frac{\{\mathbf{x},\mathbf{y}\}_\Omega}{\|\{\mathbf{x},\mathbf{y}\}_\Omega\|}\exp\!\left(-\frac{\|\{\mathbf{x},\mathbf{y}\}_\Omega\|}{L_r}\right),
\end{equation}
with $L_r$ being the interaction length scale.

Leaders are directly controlled and not subject to diffusion, while followers are driven indirectly through~\eqref{eq:repulsive_kernel} and subject to stochastic perturbations.

The leader control $\mathbf{u}$ is designed following~\cite[Sec.~VI]{maffettone2025leader-follower} to minimize the $\mathcal{L}^2$ distance between $\rho_F$ and a desired von Mises profile $\bar{\rho}_F(\mathbf{x}) = \exp\!\left(k_1\cos(x_1)+k_2\cos(x_2)+\cos(x_1-x_2)\right)$ with $k_1=k_2=9$, chosen to confine $99\%$ of followers within a goal region of radius $r^*=1$ centered in the origin. Safety is enforced for both populations via~\eqref{eq:qp_general} with barrier thresholds $\varepsilon_L=0.013$ and $\varepsilon_F=0.01$.

Simulations use $N_L=100$ leaders, $N_F=720$ followers, $D=0.005$, $T=100$, $\delta t=0.01$, 
kernel bandwidth $\sigma=0.25$, and five dangerous disks of radius $r_O = 0.5$ placed outside the goal region. 
The centers of the disks are at least a distance $\lambda = 0.75$ from the boundary of the goal region, and the minimum distance between the boundaries of any two disks is $\lambda_O = 0.25$. 
Both populations are uniformly initialized in $\Omega$, excluding a ball of radius $\lambda_r = 0.75$ around each disk center, so that the barrier constraints are satisfied at the initial time.

Fig.~\ref{fig: N_int_lf} reports the average fraction of agents entering the dangerous regions. Without MF-CBF, leader penetration stabilizes above $12\%$ and follower penetration peaks above $4\%$. With MF-CBF, both remain close to zero. 
When MF-CBF is enforced in the presence of dangerous areas, the fraction of followers  within the goal region approaches 
$0.969 \pm 0.046$, while in the nominal case without MF-CBF it reaches 
$0.999 \pm 0.001$. This confirms that the convergence error remains bounded 
when the dangerous and goal regions do not overlap. 

The initial and final 
configurations are shown in Figs.~\ref{fig:initial} and \ref{fig:final}, respectively. Specifically, at final time $T$, the fraction of followers within the goal region is $0.967$.

\section{Conclusions}\label{sec:conclusions}
In this work, we extended the MF-CBF framework to advection--diffusion systems describing the macroscopic evolution of large populations of stochastic agents. The proposed formulation enforces safety constraints at the density level through an optimization problem, enabling deterministic analysis despite the stochastic microscopic dynamics. We showed that, under MF-CBF corrections, global asymptotic convergence in the absence of safety constraints implies bounded tracking error. To construct density-level safety constraints, we introduced the kernel-weighted overlap, a functional measuring the interaction between densities via a positive-definite kernel.

The framework was validated in two representative scenarios, coverage and shepherding control, in the presence of static dangerous regions modeled through a spatial density. Numerical simulations show that MF-CBFs significantly reduce the overlap between the agent density and dangerous regions while preserving the desired macroscopic behavior.

Future work will focus on deriving analytical bounds on the asymptotic tracking error and characterizing their dependence on the geometry of the dangerous regions and the MF-CBF-induced modification of the velocity field. Another promising direction is integrating MF-CBF constraints within optimal mean-field control formulations~\cite{ruthotto2020machine, chow2023numerical}. Moreover, we plan to integrate our framework in novel continuum models for decision making \cite{lama2025nonreciprocal}, which can be recast into advection-diffusion equations.

\section*{Code and Video Availability}
The code used to generate the results presented in this paper, together with the simulation video, is publicly available at:
\href{https://github.com/SINCROgroup/Mean-field-control-barrier-functions-for-stochastic-multi-agent-systems/tree/main}{https://github.com/SINCROgroup/Mean-field-control-barrier-functions-for-stochastic-multi-agent-systems/tree/main}.

\section*{Acknowledgements}
AI tools (ChatGPT) were used to assist with the revision and  grammar correction of the manuscript. All technical content,
including theoretical results, proofs, and numerical simulations, was
produced entirely by the authors.

\bibliographystyle{IEEEtran}

\begin{thebibliography}{10}
\providecommand{\url}[1]{#1}
\csname url@samestyle\endcsname
\providecommand{\newblock}{\relax}
\providecommand{\bibinfo}[2]{#2}
\providecommand{\BIBentrySTDinterwordspacing}{\spaceskip=0pt\relax}
\providecommand{\BIBentryALTinterwordstretchfactor}{4}
\providecommand{\BIBentryALTinterwordspacing}{\spaceskip=\fontdimen2\font plus
\BIBentryALTinterwordstretchfactor\fontdimen3\font minus
  \fontdimen4\font\relax}
\providecommand{\BIBforeignlanguage}[2]{{%
\expandafter\ifx\csname l@#1\endcsname\relax
\typeout{** WARNING: IEEEtran.bst: No hyphenation pattern has been}%
\typeout{** loaded for the language `#1'. Using the pattern for}%
\typeout{** the default language instead.}%
\else
\language=\csname l@#1\endcsname
\fi
#2}}
\providecommand{\BIBdecl}{\relax}
\BIBdecl

\bibitem{d2023controlling}
R.~M. D'Souza, M.~di~Bernardo, and Y.-Y. Liu, ``Controlling complex networks
  with complex nodes,'' \emph{Nature Reviews Physics}, vol.~5, no.~4, pp.
  250--262, 2023.

\bibitem{sinigaglia2022density}
C.~Sinigaglia, A.~Manzoni, and F.~Braghin, ``Density control of large-scale
  particles swarm through {PDE}-constrained optimization,'' \emph{IEEE
  Transactions on Robotics}, vol.~38, no.~6, pp. 3530--3549, 2022.

\bibitem{elamvazhuthi2023density}
K.~Elamvazhuthi and S.~Berman, ``Density stabilization strategies for
  nonholonomic agents on compact manifolds,'' \emph{IEEE Transactions on
  Automatic Control}, vol.~69, no.~3, pp. 1448--1463, 2024.

\bibitem{ames2019control}
A.~D. Ames, S.~Coogan, M.~Egerstedt, G.~Notomista, K.~Sreenath, and P.~Tabuada,
  ``Control barrier functions: Theory and applications,'' in \emph{Proceedings
  of the European Control Conference (ECC)}, 2019, pp. 3420--3431.

\bibitem{clark2021control}
A.~Clark, ``Control barrier functions for stochastic systems,''
  \emph{Automatica}, vol. 130, p. 109688, 2021.

\bibitem{glotfelter2017nonsmooth}
P.~Glotfelter, J.~Cort{\'e}s, and M.~Egerstedt, ``Nonsmooth barrier functions
  with applications to multi-robot systems,'' \emph{IEEE Control Systems
  Letters}, vol.~1, no.~2, pp. 310--315, 2017.

\bibitem{fung2025mean}
S.~W. Fung and L.~Nurbekyan, ``Mean-field control barrier functions: A
  framework for real-time swarm control,'' in \emph{2025 American Control
  Conference (ACC)}.\hskip 1em plus 0.5em minus 0.4em\relax IEEE, 2025, pp.
  323--329.

\bibitem{gao2026banach}
X.~Gao, G.~Pascual, S.~Brown, and S.~Mart{\'\i}nez, ``Banach control barrier
  functions for large-scale swarm control,'' \emph{arXiv preprint
  arXiv:2602.05011}, 2026.

\bibitem{maffettone2025leader-follower}
G.~C. Maffettone, A.~Boldini, M.~Porfiri, and M.~di~Bernardo,
  ``Leader–follower density control of spatial dynamics in large-scale
  multiagent systems,'' \emph{IEEE Transactions on Automatic Control}, vol.~70,
  no.~10, pp. 6783--6798, 2025.

\bibitem{lama2024shepherding}
A.~Lama and M.~di~Bernardo, ``Shepherding and herdability in complex multiagent
  systems,'' \emph{Physical Review Research}, vol.~6, no.~3, p. L032012, 2024.

\bibitem{albi2016invisible}
G.~Albi, M.~Bongini, E.~Cristiani, and D.~Kalise, ``Invisible control of
  self-organizing agents leaving unknown environments,'' \emph{SIAM Journal on
  Applied Mathematics}, vol.~76, no.~4, pp. 1683--1710, 2016.

\bibitem{ascione2023mean}
G.~Ascione, D.~Castorina, and F.~Solombrino, ``Mean-field sparse optimal
  control of systems with additive white noise,'' \emph{SIAM Journal on
  Mathematical Analysis}, vol.~55, no.~6, pp. 6965--6990, 2023.

\bibitem{cortes2004coverage}
J.~Cort{\'e}s, S.~Mart{\'\i}nez, T.~Karatas, and F.~Bullo, ``Coverage control
  for mobile sensing networks,'' \emph{IEEE Transactions on Robotics and
  Automation}, vol.~20, no.~2, pp. 243--255, 2004.

\bibitem{dorigo2021swarm}
M.~Dorigo, G.~Theraulaz, and V.~Trianni, ``Swarm robotics: Past, present, and
  future [point of view],'' \emph{Proceedings of the IEEE}, vol. 109, no.~7,
  pp. 1152--1165, 2021.

\bibitem{siri2021freeway}
S.~Siri, C.~Pasquale, S.~Sacone, and A.~Ferrara, ``Freeway traffic control: A
  survey,'' \emph{Automatica}, vol. 130, p. 109655, 2021.

\bibitem{kakalis2008robotic}
N.~M. Kakalis and Y.~Ventikos, ``Robotic swarm concept for efficient oil spill
  confrontation,'' \emph{Journal of Hazardous Materials}, vol. 154, no. 1-3,
  pp. 880--887, 2008.

\bibitem{lasry2007mean}
J.-M. Lasry and P.-L. Lions, ``Mean field games,'' \emph{Japanese Journal of
  Mathematics}, vol.~2, no.~1, pp. 229--260, 2007.

\bibitem{agrawal2022random}
S.~Agrawal, W.~Lee, S.~W. Fung, and L.~Nurbekyan, ``Random features for
  high-dimensional nonlocal mean-field games,'' \emph{Journal of Computational
  Physics}, vol. 459, p. 111136, 2022.

\bibitem{lin2021alternating}
A.~T. Lin, S.~W. Fung, W.~Li, L.~Nurbekyan, and S.~J. Osher, ``Alternating the
  population and control neural networks to solve high-dimensional stochastic
  mean-field games,'' \emph{Proceedings of the National Academy of Sciences of
  the United States of America}, vol. 118, no.~31, p. e2024713118, 2021.

\bibitem{axler2020measure}
S.~Axler, \emph{Measure, integration \& real analysis}.\hskip 1em plus 0.5em
  minus 0.4em\relax Springer, 2020.

\bibitem{heinonen2001lectures}
J.~Heinonen, \emph{Lectures on analysis on metric spaces}.\hskip 1em plus 0.5em
  minus 0.4em\relax Springer, 2001.

\bibitem{gretton2012kernel}
A.~Gretton, K.~M. Borgwardt, M.~J. Rasch, B.~Sch{\"o}lkopf, and A.~Smola, ``A
  kernel two-sample test,'' \emph{Journal of Machine Learning Research},
  vol.~13, no.~1, pp. 723--773, 2012.

\bibitem{del2013mean}
P.~Del~Moral, \emph{Mean field simulation for {M}onte {C}arlo integration
  (Monographs on Statistics and Applied Probability)}.\hskip 1em plus 0.5em
  minus 0.4em\relax Boca Raton, FL: Chapman \& Hall/CRC, 2013.

\bibitem{punzo2025decentralized}
C.~Punzo, I.~Napolitano, C.~Tomaselli, and M.~di~Bernardo, ``Decentralized
  shepherding of non-cohesive swarms through cluttered environments via deep
  reinforcement learning,'' \emph{arXiv preprint arXiv:2511.21405}, 2025.

\bibitem{tomaselli_icra2026}
C.~Tomaselli, S.~Covone, R.~Andreagiovanni, and M.~di~Bernardo, ``Multi-robot
  obstacle-aware shepherding of non-cohesive target agents,'' in
  \emph{Proceedings of the IEEE International Conference on Robotics and
  Automation (ICRA)}, 2026, accepted for publication.

\bibitem{di2025continuification}
B.~Di~Lorenzo, G.~C. Maffettone, and M.~Di~Bernardo, ``A continuification-based
  control solution for large-scale shepherding,'' \emph{European Journal of
  Control}, vol.~86, p. 101324, 2025.

\bibitem{ruthotto2020machine}
L.~Ruthotto, S.~J. Osher, W.~Li, L.~Nurbekyan, and S.~W. Fung, ``A machine
  learning framework for solving high-dimensional mean field game and mean
  field control problems,'' \emph{Proceedings of the National Academy of
  Sciences of the United States of America}, vol. 117, no.~17, pp. 9183--9193,
  2020.

\bibitem{chow2023numerical}
Y.~T. Chow, S.~W. Fung, S.~Liu, L.~Nurbekyan, and S.~Osher, ``A numerical
  algorithm for inverse problem from partial boundary measurement arising from
  mean field game problem,'' \emph{Inverse Problems}, vol.~39, no.~1, p.
  014001, 2023.

\bibitem{lama2025nonreciprocal}
A.~Lama, M.~Di~Bernardo, and S.~H. Klapp, ``Nonreciprocal field theory for
  decision-making in multi-agent control systems,'' \emph{Nature
  Communications}, vol.~16, no.~1, p. 8450, 2025.

\end{thebibliography}

\end{document}